\documentclass[smallextended]{svjour3-arxiv}     

\usepackage{amsmath}

\usepackage{amsthm}

\smartqed  
\usepackage{graphicx}
\usepackage{wrapfig}
\usepackage{url}
\newcommand\eat[1]{}
\usepackage{tikz}
\usepackage{booktabs}  
\usetikzlibrary{arrows}
\usepackage{tablefootnote}
\usepackage{slashbox}
	
        \journalname{}

\usepackage{array,xspace,multirow,hhline,graphicx,xcolor,tikz,colortbl,tabularx,amsmath,amssymb,amsfonts}
\usepackage[round]{natbib}
\usetikzlibrary{shapes}
\usepackage{algorithm,algorithmic}
\usepackage{mathrsfs}
\usepackage{enumitem}
\setenumerate[1]{label=(\emph{\roman{*}}),ref=(\emph{\roman{*}}),leftmargin=*}

\usepackage{txfonts}
\usepackage{eucal} 

\newcommand{\pref}{\succsim\xspace}

\usepackage{tikz}

\newcommand{\midd}{\mathbin{: }}

\usepackage{mathtools}

	\newcommand{\Pref}[1][]{
		\ifthenelse{\equal{#1}{}}{\mathrel \succsim}{\mathop{R_{#1}}}
	}    
				                                      
	\newcommand{\sPref}[1][]{                  
		\ifthenelse{\equal{#1}{}}{\mathrel \succ}{\mathop{P_{#1}}}
	}                                          
	\newcommand{\Indiff}[1][]{                 
		\ifthenelse{\equal{#1}{}}{\mathrel \sim}{\mathop{\sim_{#1}}}
	}
	\newcommand{\prefset}[1][]{\ifthenelse{\equal{#1}{}}{\mathcal{R}}{\mathcal{R}_{#1}}}

\let\enumtemp=\enumerate
\def\enumerate{\enumtemp\itemsep 1pt}
\let\itemtemp=\itemize
\def\itemize{\itemtemp\itemsep 1pt}

\newcommand{\Omit}[1]{}

\DeclareMathOperator*{\argmax}{arg\,lex\,max}

\begin{document}


		\title{Strategyproof Multi-Item Exchange Under\\ Single-Minded Dichotomous Preferences}

	\author{Haris Aziz}

	\institute{%
	   UNSW Sydney and Data61, CSIRO,
	 	  Sydney 2052 , Australia \\
	 	  Tel.: +61-2-8306\,0490 \\
	 	  Fax: +61-2-8306\,0405 \\
	 \email{haris.aziz@unsw.edu.au}
}

	\newlength{\wordlength}
	\newcommand{\wordbox}[3][c]{\settowidth{\wordlength}{#3}\makebox[\wordlength][#1]{#2}}
	\newcommand{\mathwordbox}[3][c]{\settowidth{\wordlength}{$#3$}\makebox[\wordlength][#1]{$#2$}}
    		\renewcommand{\algorithmicrequire}{\wordbox[l]{\textbf{Input}:}{\textbf{Output}:}} 
    		 \renewcommand{\algorithmicensure}{\wordbox[l]{\textbf{Output}:}{\textbf{Output}:}}


\date{}

\maketitle

	\begin{abstract}
We consider multi-item exchange markets in which agents want to receive one of their target bundles of resources. The model encompasses well-studied markets for kidney exchange, lung exchange, and multi-organ exchange. We identify a general and sufficient condition called weak consistency for the exchange mechanisms to be strategyproof even if we impose any kind of distributional, diversity, or exchange cycle constraints. Within the class of weakly consistent and strategyproof mechanisms, we highlight two important ones that satisfy constrained Pareto optimality and strong individual rationality. Several results in the literature follow from our insights. We also  derive impossibility results when constrained Pareto optimality is defined with respect to more permissive individual rationality requirements.
	\end{abstract}



	\section{Introduction}
	
Clearing complex exchange markets are one of the successful applications of algorithmic economics and multi-agent systems~(see e.g., \citep{SoUn10a}). Key domains for which algorithms have been designed include housing markets and kidney exchange markets. In recent years, market designers are turning their attention to liver and lung exchange markets. There are proposals to explore the efficiency gains via multi-organ exchange markets~(see e.g., \citep{DiSaS17a}). Apart from exchanging organs or housing, several new digital platforms have come up that facilitate bartering of goods. Exchange markets are also useful to model time-bank scenarios in which people exchange services rather than items. 

In most of the organ exchange markets, it is supposed that agents have single-minded dichotomous preferences over outcomes: they are either satisfied with an allocation or they are not. 
We consider a general model of item or organ exchange that captures all such organ markets. Each agent is endowed with a set of items and each agent has a set of target sets of items. An agent is \emph{satisfied} if she gets any one of those target sets of items.

The primary focus both in theory and practice is on the optimisation problem of  satisfying the maximum number of agents. The problem is computationally complex even for single-unit allocations if there is a bound on the size of the exchange cycles~\citep{ABS07a}. It is also computationally hard if there are agents who have multi-unit demands as is the case of needing two liver or lung components~\citep{LuTa15a}. In this paper, we do not pursue computational issues but focus on the issue of strategyproofness of mechanisms.

In exchange settings, strategyproofness is an important property to ensure the stability and efficiency of a market. Manipulable mechanisms may give incentives for agents to hide their resources which can affect the efficiency of the market. They can also give incentives to agents to misreport their preferences in which case the mechanisms may be optimising on the wrong input. In contrast to traditional voting settings, in which strategyproofness and Pareto optimality be achieved by implementing some variant of serial dictatorship, the same argument need not work for exchange settings in which agents have the ability to `veto' certain outcomes because individual rationality is imposed.

Our primary contribution is formalizing a general property of mechanisms called \emph{weak consistency} and proving that any mechanism satisfying weak consistency is strategyproof. Within the class of weakly consistent mechanisms, we highlight two subclasses of mechanisms  which have proved useful in restricted domains. The first one is called \emph{constrained priority (CP) }and the second one is called \emph{constrained utilitarian priority  (CUP)}. Both mechanisms run a serial dictatorship type priority mechanism on the set of feasible and individually rational allocations.
We show that the mechanisms satisfy strategyproofness and constrained Pareto optimality. We then show how a subtle difference in imposing a different version of individual rationality results in impossibility results even for single-minded agents. 
	
	\section{Related Work}

The design of individually rational, Pareto efficient and strategyproof algorithms for exchange problems is one of the most intensely studied topics in market design. Under strict preferences, single-unit demands and single-unit allocations, the well-known Gale's Top Trading Cycles algorithm satisfies all the three `gold standard' properties. If the preferences are allowed to be weak, there exist  extensions of the TTC algorithm that satisfy the three properties (see e.g., \citep{JaMa12a,Plax13a,SeSa13a}).

There has also been work on strategyproof exchange with multiple endowments. When preferences are not dichotomous and agents own more than one item, then strategyproofness, Pareto efficiency, and individual rationality are generally impossible to achieve~(see e.g., \citep{Papa07c,KQW01a,TSY14a}). Some positive results depend on the fact that agents have lexicographic preferences over items or items types. Under these assumptions, it is possible to construct extensions of TTC that achieve strategyproofness, individual rationality, and Pareto optimality (see e.g., \citep{FLS+15a,SAX17a}).
The preferences we consider do not assume any of these structural restrictions that are based on lexicographic or conditional lexicographic comparisons. 
\citet{BKP15a} examined a discrete exchange setting in which agents have strict preferences over objects but there can be multiple copies of objects. They study the conditions under which strategyproofness can be achieved.

There has also been work on multi-item exchanges when preferences are dichotomous. \citet{RSU05a} show that for pairwise kidney exchanges, a priority mechanism is  strategyproof. In a highly insightful paper, \citet{Hatf05a} extended the argument to any kidney exchange framework with any kind of constraints on the feasible allocation including bounds on the exchange cycles. More recently, \citet{LuTa15a} considered lung exchange and 
again showed that a priority mechanism applied to the set of allocations satisfying the maximum number of agents is strategyproof. \citet{ZHM15a} considered general exchange problems in which agents partition the items into liked and disliked. Agents' utilities are not single-minded because they are interested in maximizing the number of liked items. 
For their model, \citet{ZHM15a} presented an algorithm that is strategyproof and Pareto optimal. However, the algorithm relies on the fact that there are no constraints such as on the size of exchange cycles. 

When preferences are single-minded, computing an allocation that satisfies the maximum number of agents is an NP-hard problem even if each agent has two items~\citep{LuTa15a}. Even for kidney exchange in which each agent gets at most one item, computing an allocation that satisfies the maximum number of agents is an NP-hard problem  if there is an upper bound on the size of the exchange cycles~\citep{ABS07a}. In this paper, we do not focus on computational problems.

Dichotomous preferences are well-studied in other context as well such as two-sided matching~\citep{BoMo04a}, auctions~\citep{MiRo13a}, non-cooperative games~\citep{HHMW01a}, single-winner voting~\citep{BrFi07c}, committee voting~\citep{ABC+16a} and probabilistic voting~\citep{BMS05a}.


%
%

\section{Model and Concepts}

An exchange market is a tuple $I=(N,O,e,D)$ where  $N=\{1,\ldots, n\}$ be a set of $n$ agents and $O$ be the set of items.  
 The vector $e=(e_1,\ldots, e_n)$ specifies the endowment $e_i\subset O$ of each agent $i\in N$. We suppose that $\bigcup_{i\in N} e_i=O$ and $e_i\cap e_j=\emptyset$ for all $i, j \in N$ such that $i\neq j$.
. Each agent has a demand set $D_i\subset 2^{2^O}$. Each element of $D_i$ is a bundle of items that is acceptable to agent $i$ and meets her goal of being in the exchange market. We say that $I'=(N,O,e',D')$ is \emph{more constrained} than $I=(N,O,e,D)$ if $e_i'\subseteq e_i$ for all $i\in N$ and $D_i'\subseteq D_i$ for all $i\in N$. We will write $I'\leq I$ if $I'$ is more constrained than $I$.


 Just like the endowment, any allocation $x=(x_1,\ldots, x_n)$ specifies the allocated bundle $x_i\subset O$ of each agent $i\in N$. In any allocation $x$, 
$x_i\cap x_j=\emptyset$  for all $i, j \in N$ such that $i\neq j$. Where the context is clear, we refer to the allocation bundle $x_i$ as the allocation of agent $i$.
We say that allocation $x$ \emph{satisfies} agent $i$ if $x_i\supseteq d$ for some $d\in D_i$.

For any two allocations $x_i,y_i\subseteq O$, one can define the preference relation $\pref_i$ of agent $i$ where $x_i \pref_i y_i$ if and only if $y_i$ satisfies $i$ implies that $x_i$ satisfies $i$. The weak preference relation gives rise to the indifference relation $\sim_i$ which holds if  $x_i \pref_i y_i$ and  $y_i \pref_i x_i$. It also gives rise to the strict part $\succ_i$ of the relation where $x_i \succ_i y_i$ if $x_i \pref_i y_i$ and $x_i \not\pref_i y_i$.

\begin{example}
	Consider an exchange market in which there are 3 agents. Agent $1$ has 4 cars $\{c_1,c_2, c_3,c_4\}$. Agent 2 has a valuable painting $p$. Agent $3$ has a sports bike $s$ and a helmet $h$. Agent 1 will only be satisfied if she gets a painting and keeps one of the cars. Agent $2$ will only be satisfied if she gets car $c_1$ or sports bike $s$ and helmet $h$. Agent $3$ will be satisfied if she gets car $c_4$. The information can be captured in the form of an exchange market $(N,O,e,D)$
	where 
	\begin{itemize}
		\item $N=\{1,2,3\}$
		\item $O=\{c_1,c_2, c_3,c_4,h, p, s\}$
		\item $e_1=\{c_1,c_2, c_3,c_4\}$; $e_2=\{p\}$; $e_3=\{s,h\}$
		\item $D_1=\{\{c,p\}\midd c\in e_1\}$; $D_2=\{\{c_1\}, \{s,h\}\}$; $D_3=\{c_4\}$.
	\end{itemize}
	   A possible allocation that satisfies all the agents is $x$ where $x_1=\{p,c_1,c_2,c_3\}; x_2=\{s,h\}; x_3=\{c_4\}$.
	\end{example}
    
 An allocation $x$ is \emph{strongly individually rational (S-IR)} if $x_i\supseteq e_i$ or $x_i\supseteq d$ for some $d\in D_i$. 
 Informally, if an agent's goal is not met, she is not interested in using up any of her endowed resources. This is a standard assumption is settings such as kidney exchange. Except for the last section, we will restrict our attention to S-IR allocations. 
Enforcing the S-IR requirement can also be seen as enforcing the standard individual rationality requirement for agents that have a special type of 
\emph{ trichotomous} preferences. In such special trichotomous preferences,  agents ideally want to get an acceptable allocation, then they next prefer getting their endowment, and they least prefer a bundle which does not satisfy them and which is not their endowment. Enforcing the S-IR requirement can also be viewed as a special type of feasibility constraint.   
 
Our model can also model altruistic donors who are happy with any allocation. We are interested in maximizing the number of satisfied agents. 
We say that an allocation $x$ is \emph{constrained Pareto optimal} within the set of allocations satisfying set $\tau$ of properties if there exists no allocation $y$ satisfying $\tau$  such that $y_i \pref_i x_i$ for all $i\in N$ and $y_i \succ_i x_i$ for some $i\in N$.
    
\paragraph{The Generality of the Model}

Our model captures any kind of exchange market in which agents are interested in getting one of the acceptable bundles. 
If $|e_i|=1$ for each $i\in N$ and $|d|=1$ for each $d\in D_i$, the market can model a kidney exchange market~\citep{ABS07a}.
If $|e_i|=2$ for each $i\in N$ and $|d|=2$ for each $d\in D_i$, the market can model a lung exchange market~\citep{LuTa15a}. It can also model multi-organ exchange market~\citep{DiSaS17a}.
The model also allows for altruistic donors who have real items to give but in return, they are happy to get null items. 
 We allow $O$ to contain dummy or null items so as allow imbalanced exchanges in which an agent may get different numbers of items that she gives away. In that case, any allocation for $I=(N,O,e,D)$ has a one-to-one correspondence with a set of cycles in which each item points to her agent and each agent points to an item she gets in the allocation. The cycles are referred to as \emph{exchange cycles}.


 \paragraph{Mechanisms}
 
    A mechanism $M$ is a function that maps each problem instance $I=(N,O,e,D)$ into an allocation.  We say that a mechanism is S-IR if it returns an S-IR allocation for each instance.

 
 
 Let $\rho(I)$ denote the set of allocations that are feasible with respect to some feasibility constraints denoted by $\rho$. Depending on what the feasibility constraints $\rho$ capture, the term $\rho(I)$ could denote the set of allocations based on pairwise exchanges or exchanges cycles of size at most 3, or satisfying some distribution or diversity constraint. In principle, $\rho$ could also capture the S-IR requirement. 
 Whatever feasibility conditions we assume on the model, we will assume that they allow the endowment allocation to be a feasible allocation. That is the only minimal assumption we make on the feasibility constraints. It ensures that there always exists at least one allocation satisfying the feasibility constraints.


Since we are interested in mechanisms that satisfy any given  feasibility constraints represented by $\rho$, we can view a mechanism as a function that maps the set of $\rho$-feasible allocations to some $\rho$-feasible allocation. 
Without loss of generality, we can think of mechanisms as choice functions that choose one of the `best' allocations from the set of feasible allocations. More generally, we will refer to $M(I,Y)$ as forcing $M$ to return an allocation from the set of allocations $Y\subseteq \rho(I)$. 


    Whereas each agent's private information is $(e_i,D_i)$, she can report any $e_i'$ and $D_i'$. We will assume $i$ can only declare an endowment that she actually has, therefore any reported endowment $e_i$, we assume that $e_i'\subseteq e_i$. We say that a mechanism is \emph{strategyproof} if there exists no instance under which some agent has an incentive to misreport her private information to get a more preferred outcome for herself. For strategyproof mechanisms, truth-telling is a dominant strategy for each agent.


	     \section{Weak consistency implies strategyproofness}


%


We first define weak consistency. Recall that in our setting, for any agent $i$ with demand set $D_i$, and for any two allocations $x_i$ and $y_i$, it is the case that $x_i \sim_i y_i$ if both $x_i$ and $y_i$ satisfy $i$ with respect to $D_i$ or neither satisfy her. 
\newpage
Consider $\rho(I)$ a set of feasible allocations defined with respect to some feasibility constraints $\rho$. 
A mechanism  $M$ is \emph{weakly consistent} if for all $I=(N,O,e,D)$, 

\begin{quote}
for any $I'\leq I$, for any $Y,Y'$ such that $Y'\subseteq Y\subseteq  \rho(I)$, and $x=M(I,{Y})$, if there exists some feasible allocation $y\in {Y'}\subset {Y}$ under instance $I'$ such that $x_i\sim_i y_i$ for all $i\in N$, then $M(I',{Y'})=z$ where $z_i \sim_i x_i$ for all $i\in N$. 
\end{quote}

The preference relation $\sim_i$ used in the definition is with respect to instance $I$. Informally, weak consistency requires that if the set of feasible outcomes is more constrained but it is still possible for each agent to get the same satisfaction level, then each agent does get the same satisfaction level from the returned outcome. 

Note that weak consistency is weaker than the consistency condition considered by \citet{Hatf05a}. In our setup, consistency can be written as follows. 
A mechanism  $M$ is \emph{consistent} if for all $I=(N,O,e,D)$, 
\begin{quote}
for any $I'\leq I$, for any $Y,Y'$ such that $Y'\subseteq Y\subseteq  \rho(I)$, and $x=M(I,{Y})$, if $x$ is a member of $Y'$ and it is feasible allocation under $I'$, then $M(I',{Y'})=x$.
\end{quote}

Weak consistency condition is a weaker form of Sen's $\alpha$ condition~\citep{Sen71a} or Maskin monotonicity~\citep{Mask99a}.



\begin{theorem}\label{th:cons-sp}
Any weakly consistent and S-IR mechanism is strategyproof. 
	\end{theorem}
		     \begin{proof}
			     Consider a weakly consistent mechanism $M$.
Suppose that an agent $i$ reports $(e_i,D_i)$ and gets an allocation $x_i$. If $x_i\supseteq d$ for some $d\in D_i$, then agent $i$ has no incentive to misreport. 
Therefore, we suppose for contradiction that $x_i\not\supseteq d$ for each $d\in D_i$.

Suppose that $i$ reports $(e_i', D_i')$ such that the resultant allocation is $x'$
and $i$ gets an acceptable bundle $b$ according to $D_i$. 
Since we assumed in our setup that $i$ can only declare an endowment that she actually has, therefore $e_i'\subseteq e_i$.
We denote the set of resources given by $i$ to the market by $s$. Since $M$ is weakly consistent, if $i$ only reports bundle $b$ as acceptable and only reports $e_i'$, then the mechanism should still choose the same allocation for agent $i$. Formally, if $i$ reports $e_i'=s$ and $D_i'=\{b\}$, then the allocation for $i$ is $b$.

Since $M$ is weakly consistent and S-IR, we can view $M$ as a choice mechanism choosing from the set of feasible and S-IR allocations. It chose an allocation $x'$ when $i$ reported  $(e_i', D_i')$.
Now consider $i$ expanding her reported endowment $e_i'$ to $e_i\supseteq e_i'$ and her reported demand set $D_i'$ to $D_i\supseteq D_i'$. 
 When $i$ reports $(e_i, D_i)$, some additional feasible and S-IR allocations may emerge as well. When $i$ reports $(e_i, D_i)$, the only \emph{new} feasible S-IR allocations are those in which $i$ gets one of her acceptable bundles according to $D_i$. By weak consistency, 
 one of the three cases hold:
 \begin{enumerate}
 	\item the same allocation $x$ is returned in which case $i$ is satisfied with respect to $D_i$
	\item some other allocation $y\neq x$ is returned such that all agents are indifferent with $x$. In this case, $i$ is satisfied with respect to $D_i$.
	\item some allocation $z$ is returned such that some agent is not indifferent between $x$ and $z$. In this case the allocation must be one of the \emph{new} allocations resulted from $i$ expanding her reported endowment $e_i'$ to $e_i\supseteq e_i'$ and her reported demand set $D_i'$ to $D_i\supseteq D_i'$. Such a new allocation cannot be one in which $i$ does not trade any item because in that case, it would even be feasible, when $i$ did not expand her reported sets $e_i'$ and $D_i$. Therefore, it must be the case that $i$ \emph{did} trade some items. Since we assume the S-IR requirement, if $i$ did trade, then $i$ gets some allocation that satisfies her with respect to one of her demands in $D_i$. 
 \end{enumerate}
 
Since $i$ is satisfied with respect to $D_i$ in all the cases, it contradicts the assumption that $i$ benefited from misreporting $(e_i', D_i')$. 
	\end{proof}

	Our proof technique is along the same lines as Theorem 2 of \citet{Hatf05a} who proved that for exchange markets with dichotomous preferences and single-unit demands and endowments, any  consistent mechanism is strategyproof. 
	In contrast to the result by \citet{Hatf05a}, (1) our result applies to any exchange market where agents can have any number of items (2) it considers a property which is weaker than the consistency property used by him and (3) it explicitly uses the S-IR requirement which was implicitly used in the proof by \citeauthor{Hatf05a}.
	As a corollary, we recover the central result from the paper by \citeauthor{Hatf05a}. Since we focus on the S-IR requirement explicitly, it leads to a better understanding when the more permissive requirement of IR is used.


	
%

	\section{Constrained Priority Mechanisms}

	In this section, we focus on two natural weakly consistent mechanisms. The mechanisms are adaptations of the idea of applying serial dictatorship and priority over the set of all feasible outcomes~(see e.g. \citep{ABB13b}). The mechanisms are parametrized with respect
to $\rho$	(a set of feasibility constraints) and $\pi$ which is a priority ordering over the agent set in which the agent in $j$-th turn is denoted by $\pi(j)$. Both rules are based on lexicographic comparisons.

	For any permutation $\pi$ of N, the CP mechanism is defined as follows. 

	\[CP(I,\rho,\pi)=\argmax_{x\in \rho(I)} (u_{\pi(1)}(x),\ldots, u_{\pi(n)}(x)). \]

	For any permutation $\pi$ of N, the CUP mechanism is defined as follows. 
	\[CUP(I,\rho,\pi)=\argmax_{x\in \rho(I)} (\sum_{i\in N}u_i(x),u_{\pi(1)}(x),\ldots, u_{\pi(n)}(x)). \]

		CP starts from the set of feasible allocations and then refines this set by using a priority ordering over the agents.
	CUP starts from the set of feasible allocations satisfying the maximum number of agents and then refines this set by using a priority ordering over the agents. Both CP and CUP are flexible enough to enforce or not enforce S-IR as part of the feasibility constraints. 
	
	The following remark points out that CUP is a general mechanism that has been used in restricted domains when S-IR is enforced.

	\begin{remark}\label{remark:equiv}
		    When applied to pairwise kidney allocation, CP and CUP are equivalent to the priority mechanism considered by \citet{RSU05a}. When applied to kidney exchange, CUP is equivalent to the  priority mechanism  considered by  \citet{Hatf05a}. When applied to kidney exchange, CUP is equivalent to the  lexicographic mechanism  considered by  \citet{LuTa15a}. When applied to house allocation with existing tenants, CUP is equivalent to the  MIR algorithm of \citet{Aziz17a}.		    
\end{remark}

Next, it is shown that CUP and CP are weakly consistent.

	\begin{lemma}\label{lemma:cup}
		CUP and CP are weakly consistent. 
		\end{lemma}
		\begin{proof}
			For any set of feasible allocations, CP and CUP can be seen as choice functions that choose the `best' allocation from the set. When that allocation is from removed the set, the next choices of CP and CUP can be selected. Hence both CP and CUP can be seen as inducing a weak order on the set of feasible allocations. 
			In such a weak order, any two feasible allocations in which each agent gets the same utility, are in the same indifference class. 
			If an allocation is removed from the set of feasible allocations, then note that the relative ordering of the other allocations does not change. Hence CP and CUP both satisfy weak consistency.
			\end{proof}


		\begin{theorem}\label{th:cup-sp}
			For any feasibility restriction $\rho$ on the set of S-IR allocations, CUP is strategyproof. 
			\end{theorem}
			\begin{proof}
			The statement follows  from Theorem~\ref{th:cons-sp} and Lemma~\ref{lemma:cup}.
			\end{proof}

			\begin{theorem}\label{th:cp-sp}
				For any feasibility restriction $\rho$ on the set of S-IR allocations, CP is strategyproof. 
				\end{theorem}
		\begin{proof}
		The statement follows from Theorem~\ref{th:cons-sp} and Lemma~\ref{lemma:cup}.
		\end{proof}

	As corollaries of the above theorem, we obtain several results where the mechanisms referred to in the corollaries are essentially applying CP to a restricted domain of exchange problems where S-IR is enforced.


	     \begin{corollary}[Theorem 1, \citet{RSU05a}]
		     For pairwise kidney exchange markets, the priority mechanism is strategyproof.
		     \end{corollary}
			     \begin{proof}
  			  The statement follows from Theorem~\ref{th:cup-sp} and Remark~\ref{remark:equiv}.
				  		     \end{proof}

	     \begin{corollary}[Corollary 3, \citet{Hatf05a}]
		     For kidney exchange markets, the priority mechanism is strategyproof even if we impose any constraint on the exchange cycle sizes.
		     \end{corollary}
			     \begin{proof}
  			  The statement follows from Theorem~\ref{th:cup-sp} and Remark~\ref{remark:equiv}.
				  		     \end{proof}
	     
		     \begin{corollary}[Theorem 2, \citet{LuTa15a}]
			     For lung exchange markets where no cycle constraints are imposed, the lexicographic mechanism \citep{LuTa15a} is strategyproof.
			     \end{corollary}
			     \begin{proof}
  			  The statement follows from Theorem~\ref{th:cup-sp} and Remark~\ref{remark:equiv}.
				  		     \end{proof}

			     \begin{corollary}[Proposition 5, \citet{Aziz17a}]
				     For house allocation with existing tenants, the MIR algorithm~\citep{Aziz17a} is strategyproof.
				     \end{corollary}
   			     \begin{proof}
     			  The statement follows from Theorem~\ref{th:cup-sp} and Remark~\ref{remark:equiv}.
   				  		     \end{proof}


If we include S-IR in the set of feasibility requirements $\rho$, then note that both CP and CUP return an allocation that is Pareto optimal allocation within the set of allocations satisfying S-IR and $\rho$.
%
Hence, we can rephrase our result in the form of the following theorem. 		
	
				     
				     \begin{theorem}\label{th:exists-sp}
	Under dichotomous preferences, for any restriction on allocations $\rho$, there exists a strategyproof mechanism that returns an allocation that satisfies constrained Pareto optimality among the set of all allocations that satisfy $\rho$ and S-IR.				     \end{theorem}
	
	Since imposing S-IR is similar to imposing IR when agents have trichotomous preferences, we can rephrase the theorem above as follows. 
	     
    				     \begin{theorem}\label{th:exists-sp}
    	Under trichotomous preferences where each agent only has her own endowment as the second preferred outcome, for any restriction on allocations $\rho$, there exists a strategyproof mechanism that returns an allocation that satisfies constrained Pareto optimality among the set of all allocations that satisfy $\rho$ and IR.				     \end{theorem}

	     \section{Impossibility Results}
	     

In the previous sections, we limited our attention to allocations that satisfy S-IR (strong individual rationality). We now explore the consequence of dropping the S-IR requirement. An allocation $x$ is \emph{individually rational (IR)} if $e_i \supseteq d$ for some $d\in D_i$ implies that  $x_i\supseteq d'$ for some $d'\in D_i$. 

Whereas IR is a less stringent requirement than S-IR, constrained Pareto optimality with respect to allocations satisfying IR is a stronger property than constrained Pareto optimality with respect to allocations satisfying S-IR. In fact, the property is stronger enough that our central results in the previous sections collapse. In particular Theorem~\ref{th:exists-sp} changes into an impossibility result as we expand the set of feasible allocations.

The proof is based on an adaptation of an impossibility result for strategyproof mechanisms maximizing the total utility of agents in which agents want to get as many desirable items as possible and their utilities are not single-minded~(Theorem 1, \citep{ZHM15a}). We adapt the argument to only consider agents having single-minded utilities: they are only satisfied if they get one of their targets sets of items. The subtle difference from the previous section is that we do not enforce S-IR which expands the set of feasible of allocations and leads to an impossibility result.

We say that an allocation is a \emph{result of pairwise exchanges} if it is a result of pairs of agents make one-for-one exchange for items and with each item changing ownership at most once. Considering the exchange cycles view of allocations as mentioned in 
Section 3, an allocation as a result of pairwise exchange is characterized by exchange cycles in which each cycle has at most two agents in it. 
 
				     \begin{theorem}\label{th:imp}
	Consider $\rho$ as the restriction of allowing allocations that are a result of pairwise item exchanges in which agents get desirable items.  
For $|N|\geq 3$, there exists no strategyproof mechanism that returns an allocation that is constrained Pareto optimal among the set of all allocations that satisfy $\rho$.	
\end{theorem}

	\begin{proof}
		Consider a set of three agents $N=\{1,2,3\}$. 
\begin{itemize}
	\item $e_1=\{a_1,a_2,a_3\}$.
	\item $e_2=\{b_1,b_2,b_3\}$.
	\item $e_3=\{c_1,c_2,c_3\}$.
\end{itemize}

Each agent $i$ likes some items $A_i\subseteq O$:
\begin{itemize}
	\item $A_1=\{b_1,b_3,c_1,c_2,c_3\}$.
	\item $A_2=\{a_1,a_3,c_1,c_2,c_3\}$.
	\item $A_3=\{a_2,a_3,b_2,b_3\}$.
\end{itemize}

\begin{itemize}
	\item $D_1=\{S\subset O \midd |S|=3 \text{ and } S\subset A_1\}$.
\item $D_2=\{S\subset O \midd |S|=3 \text{ and } S \subset A_2\}$.
				\item $D_3=\{S\subset O \midd |S|=3 \text{ and } S\subset A_3\}$.
\end{itemize}


Since we allow allocations as a result of pairwise exchanges, the resultant allocation has a corresponding matching $M$ on the node set $O$ that indicates which items are pairwise exchanged. Also, note that each agent owns items that she does not like. Each agent needs three liked items to be satisfied. Since we allow allocations as a result of pairwise exchanges, a satisfied agent need to pairwise exchange all her items with liked items. 

We first prove that in any allocation that satisfies $\rho$ and is Pareto optimal, exactly two agents are satisfied. We first prove that at most two agents can be satisfied.
Since only allocations as a result of pairwise exchanges are allowed, the maximum sized matching for 9 item nodes is at most 8. In other words, at most 8 items are exchanged so one agent still keeps an unliked item. 
Thus, at least one agent can get at most two acceptable items so it is not satisfied. 
It is also easy to see that at least two agents can be satisfied by pairwise exchanges. 

Suppose agent 1 is not satisfied. Then agent 1 can report  
$A_1'=\{b_3,c_1,c_2,c_3\}$ and 	$D_1'=\{S\subset O \midd |S|=3 \text{ and } S\subset A_1'\}$. In that case, agent 1 is not willing to get item $b_1$ in any exchange and can force this because of the requirement we have in $\rho$. Agent 2 has one item which is useless to other agents so in any $\rho$-feasible allocation it cannot be satisfied because it does not get three items from $A_2$. Therefore the only constrained Pareto optimal allocations are those in which agent 1 exchanges $a_1$ with $b_3$, and $a_2$ and $a_3$ with two the items among $e_3$. Therefore by changing her preference relation, agent $1$ gets satisfied. 

Suppose agent 2 is not satisfied. Then agent 1 can report  
$A_2'=\{a_3,c_1,c_2,c_3\}$ and 	$D_2'=\{S\subset O \midd |S|=3 \text{ and } S\subset A_2'\}$. In that case, agent 2 is not willing to get item $a_1$ in any exchange and can force this because of the requirement we have in $\rho$. Agent 1 has one item which is useless to other agents so it in any $\rho$-feasible allocation it cannot be satisfied because it does not get three items from $A_1$. Therefore the only constrained Pareto optimal allocations are those in which agent 2 exchanges $b_1$ with $a_3$, and $b_2$ and $b_3$ with two the items among $e_3$. Therefore by changing her preference relation, agent $2$ gets satisfied. 

Suppose agent 3 is not satisfied. Then agent 3 can report  
$A_3'=\{a_3,b_2,b_3\}$ and 	$D_1'=\{S\subset O \midd |S|=3 \text{ and } S\subset A_3'\}$. In that case, agent 3 is not willing to get item $a_2$ in any exchange and can force this because of the requirement we have in $\rho$. Agent 1 has one item which is useless to other agents so it in any $\rho$-feasible allocation it cannot be satisfied because it does not get three items from $A_2$. Therefore the only constrained Pareto optimal allocations are those in which agent 3 exchanges all her items with the items in the set $A_3'$. Therefore by changing her preference relation, agent $3$ gets satisfied. 
	\end{proof}
	
	The impossibility result is tight in the sense that for $|N|\leq 2$, the properties can be easily satisfied by the CP mechanism with the constraint of individual rationality. Next, we note that Theorem~\ref{th:imp} can be rephrased as follows.
	
	\begin{theorem}
	Consider $\rho$ as the restriction of allowing allocations that are a result of pairwise item exchanges in which agents get desirable items.  
	For $|N|\geq 3$, there exists no strategyproof mechanism that returns an allocation that is constrained Pareto optimal among the set of all allocations that satisfy $\rho$ and IR.		     \end{theorem}	
\begin{proof}
	Consider the proof of the previous theorem. Since none of the agents is satisfied by their endowment, the IR requirement has no bite. In particular, the set of allocations satisfying IR coincides with the set of all allocations. Hence, constrained Pareto optimality among the set of all allocations that satisfy $\rho$ is equivalent to constrained Pareto optimal among the set of all allocations that satisfy $\rho$ and IR. 
	\end{proof}	
	
	\section{Conclusions}
	
	 We have provided a clearer understanding of strategyproof multi-item exchange markets in which agents have dichotomous preferences. We focussed on scenarios where agents have single-minded preferences: they are either satisfied if they get a target bundle or they are dissatisfied if they do not get a target bundle. Under these preferences, we showed that strategyproofness can be achieved even under arbitrary restrictions on feasible allocations. In fact, there exist mechanisms that satisfy constrained Pareto optimality. We complement these results with a couple of impossibility results when we strengthen constrained Pareto optimality by enforcing a weaker notion of individual rationality. We envisage further algorithmic and game-theoretic research on the general multi-item exchange model discussed in the paper.

\section*{Acknowledgments}
Aziz gratefully acknowledges the UNSW Scientia Fellowship and Defence Science and Technology (DST). 
He thanks Barton Lee and Bahar Rastegari for useful feedback.


\begin{thebibliography}{26}
\providecommand{\natexlab}[1]{#1}
\providecommand{\url}[1]{\texttt{#1}}
\expandafter\ifx\csname urlstyle\endcsname\relax
  \providecommand{\doi}[1]{doi: #1}\else
  \providecommand{\doi}{doi: \begingroup \urlstyle{rm}\Url}\fi

\bibitem[Abbassi et~al.(2015)Abbassi, Haghpanah, and Mirrokni]{ZHM15a}
Z.~Abbassi, N.~Haghpanah, and V.~Mirrokni.
\newblock Exchange market mechanisms without money.
\newblock Technical report,
  http://people.csail.mit.edu/nima/papers/exchanges.pdf, 2015.

\bibitem[Abraham et~al.(2007)Abraham, Blum, and Sandholm]{ABS07a}
D.~Abraham, A.~Blum, and T.~Sandholm.
\newblock Clearing algorithms for barter exchange markets: Enabling nationwide
  kidney exchanges.
\newblock In \emph{Proceedings of the 8th ACM Conference on Electronic Commerce
  (ACM-EC)}, pages 295--304. ACM Press, 2007.

\bibitem[Aziz(2019)]{Aziz17a}
H.~Aziz.
\newblock Mechanisms for house allocation with existing tenants under
  dichotomous preferences.
\newblock \emph{Journal of Mechanism and Institution Design}, 2019.
\newblock Forthcoming.

\bibitem[Aziz et~al.(2013)Aziz, Brandt, and Brill]{ABB13b}
H.~Aziz, F.~Brandt, and M.~Brill.
\newblock The computational complexity of random serial dictatorship.
\newblock \emph{Economics Letters}, 121\penalty0 (3):\penalty0 341--345, 2013.

\bibitem[Aziz et~al.(2017)Aziz, Brill, Conitzer, Elkind, Freeman, and
  Walsh]{ABC+16a}
H.~Aziz, M.~Brill, V.~Conitzer, E.~Elkind, R.~Freeman, and T.~Walsh.
\newblock Justified representation in approval-based committee voting.
\newblock \emph{Social Choice and Welfare}, 48\penalty0 (2):\penalty0 461--485,
  2017.

\bibitem[Bir{\'o} et~al.(2015)Bir{\'o}, Klijn, and Papai]{BKP15a}
P.~Bir{\'o}, F.~Klijn, and S.~Papai.
\newblock Circulation under responsive preferences.
\newblock 2015.

\bibitem[Bogomolnaia and Moulin(2004)]{BoMo04a}
A.~Bogomolnaia and H.~Moulin.
\newblock Random matching under dichotomous preferences.
\newblock \emph{Econometrica}, 72\penalty0 (1):\penalty0 257--279, 2004.

\bibitem[Bogomolnaia et~al.(2005)Bogomolnaia, Moulin, and Stong]{BMS05a}
A.~Bogomolnaia, H.~Moulin, and R.~Stong.
\newblock Collective choice under dichotomous preferences.
\newblock \emph{Journal of Economic Theory}, 122\penalty0 (2):\penalty0
  165--184, 2005.

\bibitem[Brams and Fishburn(2007)]{BrFi07c}
S.~J. Brams and P.~C. Fishburn.
\newblock \emph{Approval Voting}.
\newblock Springer-Verlag, 2nd edition, 2007.

\bibitem[Dickerson and Sandholm(2017)]{DiSaS17a}
J.~P. Dickerson and T.~Sandholm.
\newblock Multi-organ exchange.
\newblock \emph{J. Artif. Intell. Res.}, 60:\penalty0 639--679, 2017.

\bibitem[Fujita et~al.(2015)Fujita, Lesca, Sonoda, Todo, and Yokoo]{FLS+15a}
E.~Fujita, J.~Lesca, A.~Sonoda, T.~Todo, and M.~Yokoo.
\newblock A complexity approach for core-selecting exchange with multiple
  indivisible goods under lexicographic preferences.
\newblock In \emph{Proceedings of the 29th AAAI Conference on Artificial
  Intelligence (AAAI)}, pages 907--913. AAAI Press, 2015.

\bibitem[Harrenstein et~al.(2001)Harrenstein, {van der Hoek}, Meyer, and
  Witteveen]{HHMW01a}
P.~Harrenstein, W.~{van der Hoek}, J.-J. Meyer, and C.~Witteveen.
\newblock Boolean games.
\newblock In J.~{van Benthem}, editor, \emph{Proceedings of the 8th Conference
  on Theoretical Aspects of Rationality and Knowledge (TARK)}, pages 287--298,
  2001.

\bibitem[Hatfield(2005)]{Hatf05a}
J.~W. Hatfield.
\newblock Pairwise kidney exchange: Comment.
\newblock \emph{Journal of Economic Theory}, 125:\penalty0 189--193, 2005.

\bibitem[Jaramillo and Manjunath(2012)]{JaMa12a}
P.~Jaramillo and V.~Manjunath.
\newblock The difference indifference makes in strategy-proof allocation of
  objects.
\newblock \emph{Journal of Economic Theory}, 147\penalty0 (5):\penalty0
  1913--1946, September 2012.

\bibitem[Konishi et~al.(2001)Konishi, Quint, and Wako]{KQW01a}
H.~Konishi, T.~Quint, and J.~Wako.
\newblock On the {S}hapley-{S}carf economy: the case of multiple types of
  indivisible goods.
\newblock \emph{Journal of Mathematical Economics}, 35\penalty0 (1):\penalty0
  1--15, 2001.

\bibitem[Luo and Tang(2015)]{LuTa15a}
S.~Luo and P.~Tang.
\newblock Mechanism design and implementation for lung exchange.
\newblock In \emph{Proceedings of the 23rd International Joint Conference on
  Artificial Intelligence (IJCAI)}, pages 209--215. AAAI Press, 2015.

\bibitem[Maskin(1999)]{Mask99a}
E.~Maskin.
\newblock Nash equilibrium and welfare optimality.
\newblock \emph{Review of Economic Studies}, 66\penalty0 (26):\penalty0 23--38,
  1999.

\bibitem[Mishra and Roy(2013)]{MiRo13a}
D.~Mishra and S.~Roy.
\newblock Implementation in multidimensional dichotomous domains.
\newblock \emph{Theoretical Economics}, 8\penalty0 (2):\penalty0 431--466,
  2013.

\bibitem[Papai(2007)]{Papa07c}
S.~Papai.
\newblock Exchange in a general market with indivisible goods.
\newblock \emph{Journal of Economic Theory}, 132:\penalty0 208--235, 2007.

\bibitem[Plaxton(2013)]{Plax13a}
C.~G. Plaxton.
\newblock A simple family of top trading cycles mechanisms for housing markets
  with indifferences.
\newblock In \emph{Proceedings of the 24th International Conference on Game
  Theory}, 2013.

\bibitem[Roth et~al.(2005)Roth, S{\"o}nmez, and {\"U}nver]{RSU05a}
A.~E. Roth, T.~S{\"o}nmez, and M.~U. {\"U}nver.
\newblock Pairwise kidney exchange.
\newblock \emph{Journal of Economic Theory}, 125:\penalty0 151--188, 2005.

\bibitem[Saban and Sethuraman(2013)]{SeSa13a}
D.~Saban and J.~Sethuraman.
\newblock House allocation with indifferences: a generalization and a unified
  view.
\newblock In \emph{Proceedings of the 14th ACM Conference on Electronic
  Commerce (ACM-EC)}, pages 803--820. ACM Press, 2013.

\bibitem[Sen(1971)]{Sen71a}
A.~K. Sen.
\newblock Choice functions and revealed preference.
\newblock \emph{Review of Economic Studies}, 38\penalty0 (3):\penalty0
  307--317, 1971.

\bibitem[Sikdar et~al.(2017)Sikdar, Adali, and Xia]{SAX17a}
S.~Sikdar, S.~Adali, and L.~Xia.
\newblock Mechanism design for multi-type housing markets.
\newblock In \emph{Proceedings of the 31st AAAI Conference on Artificial
  Intelligence (AAAI)}, pages 684--690, 2017.

\bibitem[S{\"o}nmez and {\"U}nver(2011)]{SoUn10a}
T.~S{\"o}nmez and M.~U. {\"U}nver.
\newblock Matching, allocation, and exchange of discrete resources.
\newblock In J.~Benhabib, M.~O. Jackson, and A.~Bisin, editors, \emph{Handbook
  of Social Economics}, volume~1, chapter~17, pages 781--852. Elsevier, 2011.

\bibitem[Todo et~al.(2014)Todo, Sun, and Yokoo]{TSY14a}
T.~Todo, H.~Sun, and M.~Yokoo.
\newblock Strategyproof exchange with multiple private endowments.
\newblock In \emph{Proceedings of the 28th AAAI Conference on Artificial
  Intelligence (AAAI)}, pages 805--811. AAAI Press, 2014.

\end{thebibliography}
%

\end{document}